\documentclass[11pt,oneside,reqno]{amsart}

\usepackage{amsmath}
\usepackage{xcolor}
\usepackage{amssymb}
\usepackage{amssymb,latexsym}
\usepackage{multirow}
\usepackage{graphics}

\usepackage{hyperref}
\hypersetup{
  colorlinks=true,
  linkcolor=blue,
  filecolor=magenta,    
  urlcolor=cyan,
}
\begin{document}

\theoremstyle{plain}
\newtheorem{theorem}{Theorem}[section]
\newtheorem{lemma}[theorem]{Lemma}
\newtheorem{corollary}[theorem]{Corollary}
\newtheorem{proposition}[theorem]{Proposition}
\newtheorem{question}[theorem]{Question}
\theoremstyle{definition}
\newtheorem{notations}[theorem]{Notations}
\newtheorem{notation}[theorem]{Notation}
\newtheorem{remark}[theorem]{Remark}
\newtheorem{remarks}[theorem]{Remarks}
\newtheorem{definition}[theorem]{Definition}
\newtheorem{claim}[theorem]{Claim}
\newtheorem{assumption}[theorem]{Assumption}
\numberwithin{equation}{section}
\newtheorem{example}[theorem]{Example}
\newtheorem{examples}[theorem]{Examples}
\newtheorem{propositionrm}[theorem]{Proposition}

\newcommand{\binomial}[2]{\left(\begin{array}{c}#1\\#2\end{array}\right)}
\newcommand{\fq}{{\mathbb{F}_{p^2}}}
\newcommand{\an}{{\rm an}}
\newcommand{\red}{{\rm red}}
\newcommand{\codim}{{\rm codim}}
\newcommand{\rank}{{\rm rank}}
\newcommand{\Pic}{{\rm Pic}}
\newcommand{\Div}{{\rm Div}}
\newcommand{\Hom}{{\rm Hom}}
\newcommand{\im}{{\rm im}}
\newcommand{\Spec}{{\rm Spec}}
\newcommand{\sing}{{\rm sing}}
\newcommand{\reg}{{\rm reg}}
\newcommand{\Char}{{\rm char}}
\newcommand{\Tr}{{\rm Tr}}
\newcommand{\Nr}{{\rm Nr}}
\newcommand{\res}{{\rm res}}
\newcommand{\tr}{{\rm tr}}
\newcommand{\supp}{{\rm supp}}
\newcommand{\Gal}{{\rm Gal}}
\newcommand{\Min}{{\rm Min \ }}
\newcommand{\Max}{{\rm Max \ }}
\newcommand{\Span}{{\rm Span }}

\newcommand{\Frob}{{\rm Frob}}
\newcommand{\lcm}{{\rm lcm}}
\newcommand{\ds}{\displaystyle}
\newcommand{\swt}{{\rm swt}}
\newcommand{\wt}{{\rm wt}}


\long\def\symbolfootnote[#1]#2{\begingroup%
\def\thefootnote{\fnsymbol{footnote}}\footnote[#1]{#2}\endgroup}

\newcommand{\soplus}[1]{\stackrel{#1}{\oplus}}
\newcommand{\dlog}{{\rm dlog}\,}  
\newcommand{\limdir}[1]{{\displaystyle{\mathop{\rm
lim}_{\buildrel\longrightarrow\over{#1}}}}\,}
\newcommand{\liminv}[1]{{\displaystyle{\mathop{\rm
lim}_{\buildrel\longleftarrow\over{#1}}}}\,}
\newcommand{\boxtensor}{{\Box\kern-9.03pt\raise1.42pt\hbox{$\times$}}}
\newcommand{\sext}{\mbox{${\mathcal E}xt\,$}}
\newcommand{\shom}{\mbox{${\mathcal H}om\,$}}
\newcommand{\coker}{{\rm coker}\,}
\newcommand{\ord}{{\rm ord}\,}
\renewcommand{\iff}{\mbox{ $\Longleftrightarrow$ }}
\newcommand{\onto}{\mbox{$\,\\rangle \rangle \rangle \hspace{-.5cm}\to\hspace{.15cm}$}}
\newcommand{\F}{
\mathbb{F}}
\newenvironment{pf}{\noindent\textbf{Proof.}\quad}{\hfill{$\Box$}}

\title{Polynomial representation of additive cyclic codes and new quantum codes}

\author{Reza Dastbasteh \and Khalil Shivji}
\address{Department of Mathematics, Simon Fraser University, Burnaby, BC, Canada}
\email{ rdastbas@sfu.ca, kh411@protonmail.com}

\maketitle
\begin{abstract}
We give a polynomial representation for additive cyclic codes over $\F_{p^2}$. This representation will be applied to uniquely present each additive cyclic code by at most two generator polynomials. We determine the generator polynomials of all different additive cyclic codes. A minimum distance lower bound for additive cyclic codes will also be provided using linear cyclic codes over $\F_p$. We classify all the symplectic self-dual, self-orthogonal, and nearly self-orthogonal additive cyclic codes over $\F_{p^2}$. Finally, we present ten record-breaking binary quantum codes after applying a quantum construction to self-orthogonal and nearly self-orthogonal additive cyclic codes over $\F_{4}$.
\end{abstract}


\section{Introduction}

Quantum error-correcting codes, or simply quantum codes, are used in quantum computation to protect quantum information from corruption by noise (decoherence). A general framework of quantum codes is provided in \cite{Encyclopedia,Grassl}. Throughout this paper, $\mathbb{F}_{p^2}$ is the finite field of $p^2$ elements, where $p$ is a prime number. The parameters of a quantum code over $\F_p$ that encodes $k$ logical qubits to $n$ physical qubits and has minimum distance $d$ is denoted by $[[n,k,d]]_p$.
An important family of quantum codes with many similar properties as classical block codes is the family of quantum stabilizer codes. In particular, quantum stabilizer codes are constructed using additive codes which are self-orthogonal with respect to a certain symplectic inner product. Several constructions of quantum stabilizer codes from various classical codes are given in \cite{Ketkar}.  
An interesting modification of the original definition of quantum stabilizer codes is by relaxing its self-orthogonality constraint \cite{Reza,Reza2}. This method enables us to construct good quantum codes using not necessarily self-orthogonal additive codes over $\F_4$. Previously, this modification was applied for the construction of new quantum codes from different families of linear codes \cite{RezaD,ezerman,Lisonek}.

Additive cyclic codes are of interest due to their rich algebraic properties and application in the construction of quantum codes. 
There have been several works in the literature toward the classification of additive cyclic codes for different applications \cite{Bierbrauer,Cao,Dey,Huffman, Huffman2, Samei}, and also due to their connection to other families of block codes such as quasi-cyclic codes \cite{Cem}. 
In \cite{Huffman}, a canonical decomposition of additive cyclic code over $\F_4$ was introduced using certain finite field extensions of $\F_4$. This decomposition was applied to determine self-orthogonal and self-dual additive cyclic codes over $\F_4$ with respect to the trace inner product. In \cite{Calderbank}, it was shown that each additive cyclic code over $\F_4$ of length $n$ can be generated by $\F_2$-span of at most two polynomials in $\F_4[x]/\langle x^n-1\rangle$ and their cyclic shifts. Moreover, a criterion for the self-orthogonality of such codes with respect to the trace inner product was provided. 
Another interesting construction for a subclass of additive cyclic code, namely twisted codes, was provided in \cite{Bierbrauer}. This construction is analogous to the way linear cyclic codes are constructed. In spite of many useful properties of twisted codes, all additive cyclic codes cannot be described using the theory of additive twisted codes.

In this work, we first give a canonical representation of all $\F_p$-additive cyclic codes over $\F_{p^2}$ using at most two generator polynomials. Our representation is more computationally friendly than the canonical representation of \cite{Huffman}. This representation allows us to give a minimum distance lower bound for additive cyclic codes over $\F_{p^2}$ using the minimum distance of linear cyclic codes over $\F_p$.
Moreover, we provide a unique set of generator polynomials for each additive cyclic code over $\F_{p^2}$. This representation of generator polynomials will be used to characterize all self-orthogonal and self-dual additive cyclic codes with respect to the symplectic inner product. 
We also determine the generator polynomials of the symplectic dual of a given additive cyclic code over $\F_{p^2}$, and compute nearly the self-orthogonality of each additive cyclic code using only its generator polynomials. This allows us to apply the nearly self-orthogonal construction of quantum codes developed in \cite{Reza,Reza2}. 
In particular, we provide a list of eleven record-breaking binary quantum codes after applying the mentioned quantum construction to nearly self-orthogonal additive cyclic codes. Furthermore, applying secondary constructions to our new quantum codes produce many more record-breaking binary codes. Note that such new quantum codes cannot be constructed using self-orthogonal additive cyclic codes of the same length. 

This paper is organized as follows. Section \ref{S: pre} briefly recalls the essential terminologies used in this work.
Section \ref{additive cyclic codes} gives a canonical representation of additive cyclic codes over $\F_{p^2}$. In fact, we follow a module theory approach to decompose each additive cyclic code using its polynomial representation in 
$\F_{p^2}[x]/\langle x^n-1\rangle$. 
In Section \ref{S: quantum dual}, we compute the symplectic dual of each additive cyclic code. We provide the necessary and sufficient conditions for an additive cyclic code to be self-orthogonal, self-dual, or nearly self-orthogonality with respect to the symplectic inner product.  
Finally, in Section \ref{S: new codes}, we present the parameters of our record-breaking quantum codes.

\section{Preliminaries}\label{S: pre}
Let $\omega$ be a primitive element of $\F_{p^2}$. Then the set $\{1,\omega\}$ forms a basis for $\F_{p^2}$ over $\F_p$.
Let $a+b \omega$ and $a'+b'\omega \in \F_{p^2}^n$, where $a,a',b,b'\in \F_p^n$. The {\em symplectic inner product} of $a+b \omega$ and $a'+b'\omega$ is defined by
\begin{equation}\label{E: symp}
\langle a+b\omega,a'+b'\omega \rangle_s=a' \cdot b-a \cdot b'.
\end{equation}
An $\mathbb{F}_{p}$-linear subspace $C \subseteq \mathbb{F}_{p^2}^{n}$ is called a length $n$ {\em additive code} over $\F_{p^2}$. We denote the $\F_p$-dimension of an additive code $C$ over $\F_{p^2}$ with $\dim_{\F_p}(C)$.
Let $C \subseteq \mathbb{F}_{p^2}^{n}$ be an additive code over $\F_{p^2}$ such that $\dim_{\F_p}(C)=k$. Then we call $C$ an $(n,p^k)$ code. 
The set 
\begin{center}
$C^{{\bot}_s}=\{x \in \mathbb{F}_{p^2}^{n} : \langle x,y \rangle_s=0$ for all $y \in C \}$.
\end{center}
is called the {\em symplectic dual} of $C$. One can easily see that $C^{{\bot}_s}$ is an $(n, p^{2n-k})$ additive code over $\F_{p^2}$. The code $C$ is called \textit{self-orthogonal} (respectively \textit{self-dual}) if $C\subseteq C^{\bot_s}$ (respectively if $C= C^{\bot_s}$). For each $x \in \F_{p^2}^n$, we denote the number of non-zero coordinates of $x$ by $\wt(x)$. Moreover, the minimum weight among non-zero vectors of an additive code $C$ is denoted by $d(C)$.
The connection between quantum stabilizer codes and classical additive codes was initially formulated by the independent works of Calderbank, Rains, Shor, and Sloane \cite{Calderbank} and Gottesman \cite{Gottesman}. 
A non-binary version of this connection is provided below. 

\begin{theorem}\cite[Corollary 16]{Ketkar}\label{T: quantum def}
Let $C$ be an $(n,p^{n-k})$ additive code over $\F_{p^2}$. Then there exists an $[[n,k,d]]_p$ quantum stabilizer code if $C$ is symplectic self-orthogonal, where $d=\min\{\wt(x): x \in C{^\bot_s} \setminus C\}$ if $k>0$ and $d=\min\{\wt(x): x \in C\}$ if $k=0$. 
\end{theorem}

The quantum code of Theorem \ref{T: quantum def} is called {\em pure} if $d=d(C^{\bot_s})$. There are several secondary constructions of quantum code. A short list of such constructions is provided below. 

\begin{theorem}\cite[Section XV]{Ketkar}\label{T: secondary}
Let $C$ be an $[[n,k,d]]_p$ quantum code. 
\begin{enumerate}
\item If $k>0$, then an $[[n+1,k,d]]_p$ quantum code exists.
\item If $C$ is pure and $n,d\geq 2$, then an $[[n-1,k+1,d-1]]_p$ pure quantum code exists.
\item If $k>1$, then there exists an $[[n,k-1,d]]_p$ quantum code.
\end{enumerate} 
\end{theorem}

\section{Additive cyclic codes over $\F_{p^2}$}\label{additive cyclic codes}

Throughout this section, we assume that $n$ is a positive integer such that $(n,p)=1$ and $\mathbb{F}_{p^2}=\{\alpha+\beta \omega: \alpha,\beta \in \mathbb{F}_p\}$, where $\omega$ is a root of a degree two irreducible polynomial over $\mathbb{F}_p$. 
In this section, we provide a canonical representation of additive cyclic codes over the field $\mathbb{F}_{p^2}$. In particular, we give a unique representation of each additive cyclic code over $\F_{p^2}$ using at most two generator polynomials. Moreover, we determine the generator polynomials of all different additive cyclic codes over $\F_{p^2}$. 
In particular, each additive cyclic code over $\F_p^2$ is a linear combination of cyclic shifts of its generator polynomials. 
Such representation is also suitable for practical computations of additive cyclic codes, especially using Magma computer algebra system \cite{magma}. 
More particularly, there exists a built-in function in Magma which forms additive cyclic codes generated by two given generator polynomials. At the end of this section, we give a minimum distance lower bound for the minimum distance of additive cyclic codes over $\F_{p^2}$ using the minimum distance of linear cyclic codes over $\F_p$.

\begin{definition}
An $\F_p$-subspace $C\subseteq \mathbb{F}_{p^2}^n$ is called an {\em additive cyclic code} of length $n$ over $\mathbb{F}_{p^2}$, if for every $(a_0,a_1,\dots,a_{n-1})\in C$, the vector $(a_{n-1},a_0,\dots,a_{n-2})$ is also a codeword of $C$.
\end{definition}

We will use the following concepts of module theory frequently in this section, and for more details one, for example, can see \cite[Chapter 12]{Dummit}. Let $R$ be a principal ideal domain and $M$ be an $R$-module. The {\em annihilator} of $M$ is an ideal of $R$ defined by $\{ r \in R: rm=0\ \text{for any} \ m\in M\}$. An element $m\in M$ is called a {\em torsion element}, if there exists $0\neq r \in R$ such that $rm=0$. The module $M$ is called a {\em torsion module} if all of its elements are torsion.  
The following theorem, known as the primary decomposition theorem of modules, plays an important role in our representation of additive cyclic codes. 
\begin{theorem}\label{module} \cite[Chapter $12$, Theorem $7$]{Dummit}
Let $R$ be a principal ideal domain and $M$ be a torsion $R$-module with the annihilator $\langle a \rangle \neq 0$. Let $a=u\displaystyle\prod_{i=1}^n p_i^{a_i}$, where $u$ is a unit and $p_i$ is a prime element for each $1\le i \le n$. Then we can decompose $M$ as a direct sum of its submodules in the form
\begin{equation}
M=\bigoplus_{i=1}^n N_i,
\end{equation}
where $N_i=\{x\in M:xp_i^{a_i}=0\}$ for each $1\le i \le n$. 
\end{theorem}

Each element $(a_0,a_1,\ldots,a_{n-1}) \in \F_{p^2}^n$ can be represented uniquely as a polynomial in $\F_{p^2}[x]/\langle x^n-1\rangle$ in the form $\displaystyle\sum_{i=0}^{n-1}a_ix^i$.
One can easily verify that, under this correspondence, a length $n$ additive cyclic codes over $\mathbb{F}_{p^2}$ is an $\mathbb{F}_p[x]$-submodule of $\mathbb{F}_{p^2}[x]/\langle x^n-1\rangle$. 

\begin{notation}
Let $f$ and $g \in \mathbb{F}_{p^2}[x]/\langle x^n-1\rangle$. We fix the following notations for the rest of this paper.
\begin{enumerate}
 \item The ideal generated by $f$ in $\mathbb{F}_{p^2}[x]/\langle x^n-1\rangle$ is denoted by $\langle f\rangle _{\mathbb{F}_{p^2}[x]}$. Equivalently it is the ${\mathbb{F}_{p^2}[x]}$-submodule of $\mathbb{F}_{p^2}[x]/\langle x^n-1\rangle$ generated by the polynomial $f$.
 \item The ${\mathbb{F}_p[x]}$-submodule of $\mathbb{F}_{p^2}[x]/\langle x^n-1\rangle$ generated by the polynomial $g$ is denoted by $\langle g\rangle _{\mathbb{F}_p[x]}$.
 \end{enumerate}
\end{notation}

A straightforward computation shows that the annihilator of $\mathbb{F}_{p^2}[x]/\langle x^n-1\rangle$ as an $\mathbb{F}_p[x]$-module is the ideal $\langle x^n-1\rangle$. Moreover, we can decompose $x^n-1$ over $\mathbb{F}_p[x]$ as $x^n-1=\displaystyle \prod_{i=1}^{s} f_i(x)$, where each $f_i(x)$ is an irreducible polynomial corresponding to a $p$-cyclotomic coset modulo $n$. Next, we apply Theorem $\ref{module}$ to $\mathbb{F}_{p^2}[x]/\langle x^n-1\rangle$. 
It is straightforward to see that 
\begin{equation}\label{equation1}\
\mathbb{F}_{p^2}[x]/\langle x^n-1\rangle=\bigoplus_{i=1}^s N_i,
\end{equation}
where $N_i=\langle (x^n-1)/f_i(x)\rangle _{\F_{p^2}[x]}$ for each $1\le i \le s$.  
We call a non-zero length $n$ additive cyclic code $C$ over $\F_{p^2}$ {\em irreducible} if for any additive cyclic code $D\subseteq C$, then $D=\{0\}$ or $D=C$. 
The next lemma shows that each $N_i$ can be decomposed as a direct sum of two irreducible additive cyclic codes. We determine the generator polynomial of all irreducible additive cyclic codes inside $N_i$ and provide other useful information about additive cyclic codes inside each $N_i$.

\begin{lemma}\label{case2}
 Let $f(x)$ be an irreducible divisor of $x^n-1$ over $\F_p[x]$ with $\deg(f)=k$ and $N= \langle (x^n-1)/f(x)\rangle _{\F_{p^2}[x]}$.
\begin{enumerate}
\item Let $0\neq r(x) \in N$, then the set $L=\{r(x),xr(x),\ldots, x^{k-1}r(x)\}$ forms a basis for $\langle r(x)\rangle_{\F_p[x]}$ as an $\F_p$ vector space.
\item Let $0\neq C \subsetneq N$ be an additive cyclic code. 
 The code $C$ has $\F_p$-dimension $k$ and $C=\langle r(x)\rangle _{\F_p[x]}$ for any $0\neq r(x) \in C$.
\item The additive cyclic code $N$ can be decomposed as 
$$N=\langle (x^n-1)/f(x) \rangle _{\F_p[x]} \oplus \langle \omega((x^n-1)/f(x)) \rangle _{\F_p[x]}.$$
Moreover, $\dim_{\F_p}(N)=2k$ and $N$ is linear over $\F_{p^2}$.
\item The number of irreducible additive cyclic codes inside $N$ is $2^{k}+1$. In particular, the following set gives all the different generator polynomials of such additive cyclic codes.
\begin{equation}\label{one generator form}
A=\{\big((x^n-1)/f(x)\big)\big(\omega+g(x)\big):g(x)\in \F_p[x],\deg(g(x)) < k\}\cup\{(x^n-1)/f(x)\}.
\end{equation}
\end{enumerate}
\end{lemma}

\begin{proof}
$(1)$ Obviously $L \subseteq \langle r(x)\rangle_{\F_p[x]}$. Suppose, on the contrary, that $L$ is linearly dependent over $\F_p$. Hence we can find a polynomial $0\neq s(x) \in \F_p[x]$ of degree less than $k$ such that $r(x)s(x) \equiv 0 \pmod{x^n-1}$. Since $(x^n-1)/f(x) \mid r(x)$ and $f(x)$ is irreducible, we conclude that $f(x) \mid s(x)$. However, it is a contradiction with the fact that $\deg(s(x))< k$. 
This shows that $L$ is linearly independent over $\F_p$. Note that the set $L \cup \{x^kr(x)\}$ is linearly dependent over $\F_p$ as this new set generates $f(x)r(x)\equiv0 \pmod{x^n-1}$. In a similar fashion, one can show that $\{x^ir(x)\}$ for $k<i<n-1$ can be written as a linear combination of elements of $L$ over $\F_p$.
Therefore, $L $ forms a basis for $\langle r(x)\rangle _{\F_p[x]}$.

$(2)$ Let $0\neq r(x) \in C$. Suppose in contrary that $\langle r(x) \rangle_{\F_p[x]} \subsetneq C$. Then there exists a polynomial $s(x) \in C$ such that $s(x)\not \in \langle r(x) \rangle_{\F_p[x]}$. Note that $\langle r(x)\rangle_{\F_p[x]} \cap \langle s(x)\rangle_{\F_p[x]}=\{0\} $ as otherwise, by part (1), for any polynomial $a(x)$ in the intersection, we have 
 $$ \langle r(x)\rangle_{\F_p[x]}= \langle a(x)\rangle_{\F_p[x]}= \langle s(x)\rangle_{\F_p[x]},$$
which is a contradiction. 
Thus $C = \langle r(x)\rangle_{\F_p[x]}$ and has dimension $k$ over $\F_p$.

$(3)$ It is easy to see that $\langle (x^n-1)/f(x)\rangle _{\F_p[x]} \cap \langle \omega((x^n-1)/f(x)\rangle _{\F_p[x]}=\{0\}$ and 
$$N=\langle (x^n-1)/f(x) \rangle _{\F_p[x]} \oplus \langle \omega((x^n-1)/f(x)) \rangle _{\F_p[x]}.$$
Hence $N$ has dimension $2k$ over $\F_p$. The linearity part follows immediately from the structure of its generator polynomials.

$(4)$ In order to find an additive cyclic code with $\F_p$-dimension $k$, we need to choose a nonzero polynomial $r(x) \in N$ to be its generator. Also, any non-zero elements of $\langle r(x)\rangle _{\F_p[x]}$ generates the same code. Hence the number of additive cyclic codes with one non-zero generator inside $N$ is $\frac{2^{2k}-1}{2^k-1}=2^k+1$.

Let $C_1$ and $C_2$ be two $k$-dimensional additive cyclic codes inside $N$. If $C_1\cap C_2 \neq \{0\}$, then $C_1=C_2$ by part $(1)$. Equivalently, if $C_1+C_2=N$, then $C_1\cap C_2=\{0\}$.
Now we show that different elements of the set $A$ generate different codes. Let $g(x)\in \F_p[x]$ such that $\deg(g(x)) < k$.
Clearly the additive cyclic code $C_1=\langle (x^n-1)/f(x),((x^n-1)/f(x))(g(x)+\omega)\rangle _{\F_p[x]}$ contains $(x^n-1)/f(x)$ and $\omega(x^n-1)/f(x)$. Therefore $C_1=N$. So $\langle (x^n-1)/f(x) \rangle_{\F_p[x]}$ and $\langle((x^n-1)/f(x))(g(x)+\omega)\rangle_{\F_p[x]}$ are different additive cyclic codes.

Let $g_1(x)$ and $g_2(x) \in \F_p[x]$ be two different polynomials of degree less than $k$. The code $C=\langle ((x^n-1)/f(x))(\omega+g_1(x)),((x^n-1)/f(x))(\omega+g_2(x))\rangle _{\F_p[x]}$ contains $(x^n-1)/f(x)$ and $ \omega(x^n-1)/f(x)$. It is mainly because
$$\langle \big((x^n-1)/f(x)\big)(g_1(x)-g_2(x))\rangle _{\F_p[x]}=\langle (x^n-1)/f(x)\rangle _{\F_p[x]}.$$ 
 Thus $C=N$. This implies that the additive cyclic codes $\langle ((x^n-1)/f(x))(\omega+g_1(x)) \rangle _{\F_p[x]}$ and $\langle ((x^n-1)/f(x))(\omega+g_2(x)) \rangle _{\F_p[x]}$ are different. 
 This proves that the set $A$ contains all the different generators of irreducible additive cyclic codes inside $N$.
\end{proof}

As we mentioned in part $(1)$ of Lemma \ref{case2}, each additive cyclic code inside $ \langle (x^n-1)/f(x)\rangle _{\F_{p^2}[x]}$ can have many different generator polynomials. 
Through the next remark, we fix a canonical representation for each additive cyclic code inside $N$.  

\begin{remark}\label{R: the gen}
For each additive code $0\neq C \subsetneq \langle (x^n-1)/f(x)\rangle _{\F_{p^2}[x]}$, we fix its generator polynomial inside the set $A$, introduced in $(\ref{one generator form})$, to be ``the" generator polynomial of $C$.  
Similarly, the additive cyclic code $C'=\langle (x^n-1)/f(x)\rangle _{\F_{p^2}[x]}$ can be generated by the polynomials $(x^n-1)/f(x)$ and $\omega((x^n-1)/f(x))$. We call them ``the'' generator polynomials of $C'$. 
\end{remark}
This representation helps to uniquely identify each additive cyclic code inside $N$ and avoid considering the same code more than once.
Next, we use the result of Lemma \ref{case2} and characterize all the additive cyclic codes of length $n$ over $\F_{p^2}$. Recall that $x^n-1=\ds \prod_{i=1}^{s}f_i(x)$, where $f_i(x)$ is an irreducible polynomial over $\F_p[x]$ for each $1\le i \le s$ and $N_i=\langle (x^n-1)/f_i(x)\rangle _{\F_{p^2}[x]}$.  

\begin{theorem} \label{decomposition}
Let $C$ be a length $n$ additive cyclic code over $\mathbb{F}_{p^2}$. Then
\begin{enumerate}
\item[(i)] we can decompose the code $C$ as $C=\displaystyle\bigoplus_{i=1}^{s} C_i$, where each $C_i$ is an additive cyclic code inside $N_i$.

\item[(ii)] we have $C=\langle g(x)+\omega k(x),\omega h(x)\rangle _{\F_p[x]}$, where 
\begin{enumerate}

\item $g(x)+\omega k(x)=\displaystyle\sum_{i=1}^{s}g_i(x)+\omega k_i(x)$, 
\item$h(x)=\displaystyle\sum_{i=1}^{s}h_i(x)$, 
\item and $C_i$ has the generator polynomial(s) $g_i(x)+\omega k_i(x)$ and $\omega h_i(x)$ selected as discussed in Remark \ref{R: the gen}. 
\end{enumerate}

\item[(iii)] $\dim_{\F_p}(C)=\ds\sum_{i=1}^{s} (\deg(f_i) \times \#$ of non-zero generators of $C_i$).

\end{enumerate}
\end{theorem}

\begin{proof}
(i) As we mentioned in $(\ref{equation1})$, the following decomposition holds
$$\mathbb{F}_{p^2}[x]/\langle x^n-1\rangle=\bigoplus_{i=1}^s N_i.$$ 
So we can express $C$ as $C=\bigoplus_{i=1}^{s} C_i$, where each $C_i$ is an additive cyclic codes inside $N_i$.
 
(ii) We show that the additive cyclic codes $C=\ds\bigoplus_{i=1}^{s} C_i$ and $\langle g(x)+\omega k(x),\omega h(x)\rangle _{\F_p[x]}$ are the same. First note that $ g(x)+\omega k(x),\omega h(x) \in C$ and thus $\langle g(x)+\omega k(x),\omega h(x)\rangle _{\F_p[x]} \subseteq C$. 
Let $1\le i \le s$ be a fixed integer. Since $\big((x^n-1)/f_i(x)\big) \mid g_i(x), k_i(x), h_i(x)$ and 
$$\big((x^n-1)/f_i(x)\big)g_j(x) \equiv \big((x^n-1)/f_i(x)\big)k_j(x) \equiv \big((x^n-1)/f_i(x)\big)h_j(x) \equiv 0 \pmod{x^n-1}$$
for any $j\neq i$,
we have
 $$\big((x^n-1)/f_i(x)\big)\big(g(x)+\omega k(x)\big)\equiv \big((x^n-1)/f_i(x)\big)\big(g_i(x)+\omega k_i(x)\big) \pmod{x^n-1}$$ 
and 
$$\big((x^n-1)/f_i(x)\big)\omega h(x) \equiv \big((x^n-1)/f_i(x)\big) \omega h_i(x) \pmod{x^n-1}.$$ 
Moreover, we have 
$$C_i=\langle g_i(x)+\omega k_i(x),\omega h_i(x)\rangle _{\F_p[x]}=\langle \big((x^n-1)/f_i(x)\big)\big(g(x)+\omega k(x)\big),\big((x^n-1)/f_i(x)\big)\omega h(x)\rangle _{\F_p[x]}.$$
Thus 
\begin{equation*}
C_i \subseteq \langle g(x)+\omega k(x),\omega h(x)\rangle _{\F_p[x]}.
\end{equation*}
This show that $\ds \bigoplus_{i=1}^{s} C_i \subseteq \langle g(x)+\omega k(x),\omega h(x)\rangle _{\F_p[x]}$ and completes the proof. 

(iii) Note that $\dim_{\F_p}(C)=\ds\sum_{i=1}^{s}\dim_{\F_p}(C_i)$. Moreover, by Lemmas $\ref{case2}$, $\dim_{\F_p}(C_i)=0$, $k_i$, or $2k_i$ if $C_i=0$, $C_i$ is generated by one generator polynomial, or $C_i$ has two generator polynomials, respectively. Combining these facts with the result of part (i) completes this proof.
\end{proof}

Through the next corollary, we characterize all the length $n$ irreducible additive cyclic codes over $\F_{p^2}$.
\begin{proposition}
Let $C$ be an additive cyclic code of length $n$ over $\F_{p^2}$. Then $C$ is irreducible if and only if $C=\langle r(x)\rangle _{\F_p[x]}$ for some $0 \neq r(x) \in N_i$ and $1\le i \le s$. Moreover, there are $\ds\sum_{i=1}^{s} (2^{\deg(f_i)}+1)$ many different irreducible additive cyclic codes.
\end{proposition}

\begin{proof}
Let $C=\langle r(x)\rangle _{\F_p[x]}$ for some $0 \neq r(x) \in N_i$ and $1\le i \le s$.
The result of part (1) in Lemma \ref{case2} shows that $C$ is irreducible. Conversely, let $C$ be an irreducible additive cyclic code. Then by part (i) of Theorem \ref{decomposition} we have $C=\displaystyle\bigoplus_{i=1}^{s}C_i$. Since $C$ is irreducible, we have $C=C_j$ for some $1\le j \le s$. Moreover, since $N_j$ is not irreducible by Lemma \ref{case2}
part (3), we conclude that $C=\langle r(x)\rangle _{\F_p[x]}$ for some $0 \neq r(x) \in N_j$.

Inside each $N_i$, there are $2^{\deg(f_i)}+1$ many different one generator additive cyclic codes. Hence the total number of irreducible codes is $\ds \sum_{i=1}^{s} (2^{\deg(f_i)}+1)$.
\end{proof}

\begin{remark}\label{R: gen form}
Henceforth, we always represent each additive cyclic code with its generator polynomials $g(x)+\omega k(x)$ and $\omega h(x)$ introduced in part (ii) of Theorem $\ref{decomposition}$. Moreover, the way we generate these polynomials is unique, and therefore each additive cyclic code has a unique set of generators. 
\end{remark}

From now on, we call $\F_{p^2}$-linear cyclic codes simply linear cyclic codes.
Let $C=\langle g(x)+\omega k(x),\omega h(x)\rangle _{\F_p[x]}$ be a length $n$ additive cyclic code over $\F_{p^2}$. 
Note that Theorem \ref{decomposition} and part (3) of Lemma \ref{case2} imply that $C$ is linear if and only if $g(x)=h(x)$ and $k(x)=0$. 
Hence linear cyclic codes can be easily distinguished from non-linear cyclic codes. 

 Next, we provide a minimum distance bound for additive cyclic codes using linear cyclic codes over $\F_p$. In general, the minimum distance computation for linear codes is faster than the additive codes. 
 Hence the following result can speed up the minimum distance computation for additive cyclic codes. We denote the minimum distance of a code $C$ with $d(C)$.

\begin{theorem}\label{distance bound 1}
Let $C=\langle g(x)+\omega k(x),\omega h(x)\rangle _{\F_p[x]}$ be a length $n$ additive cyclic code over $\F_{p^2}$.  
Let $G(x)=\frac{x^n-1}{\gcd(x^n-1,g(x))}$, and let $S(x)$ be the generator polynomial of the intersection of the length $n$ linear cyclic code generated by $k(x)$ and the linear cyclic code generated by $h(x)$ over $\F_p$. Suppose that $D_1$, $D_2$, $D_3$, and $D_3$ are the length $n$ linear cyclic codes
over $\F_p$ generated by $g(x)$, $\gcd(k(x),h(x))$, $\gcd(G(x)k(x),h(x))$, and $\frac{g(x)S(x)}{\gcd(x^n-1,k(x))}$, respectively. Then
\begin{equation}\label{distance}
\min\{d(D_3),d(D_4), \max\{d(D_1),d(D_2)\} \}\le d(C).
\end{equation} 
\end{theorem}

\begin{proof}
Only the following three types of codewords may appear in the code $C$.
$$T_1=\{a(x) \in C: \ 0\neq a(x)\in \F_p[x] \},$$
$$T_2=\{\omega b(x) \in C: \ 0\neq b(x)\in \F_p[x] \},$$
$$T_3=\{a(x)+\omega b(x) \in C: \ 0\neq a(x),0\neq b(x)\in \F_p[x] \}.$$
We bound the minimum distance of $C$ by considering the minimum distance in each of these sets. Let $f(x)\in T_1$. Then we can write it as $f(x)=a_1(x)(g(x)+\omega k(x))+b_1(x)\omega h(x)$ for some $a_1(x),b_1(x)\in \F_p[x]$. Hence $f(x)=a_1(x)g(x)$ and $a_1(x)k(x)+b_1(x)h(x)\equiv 0 \pmod{x^n-1}$. This implies that $a_1(x)k(x)$ is an element of the length $n$ linear cyclic code over $\F_p$ generated by $S(x)$. Hence $ \frac{S(x)}{\gcd(x^n-1,k(x))} \mid a_1(x) $. In other words, $f(x)=a(x)g(x) \in D_4$.  

Next, let $\omega f_1(x)\in T_2$. Then $\omega f_1(x)=a_1(x)(g(x)+\omega k(x))+b_1(x)\omega h(x)$ for some $a_1(x),b_1(x)\in \F_p[x]$. Then $a_1(x)g(x)\equiv 0 \pmod{x^n-1}$ or equivalently $G(x) \mid a_1(x)$. This implies that $f_1(x)=a_1(x)k(x)+b_1(x)h(x)$.
Therefore, $f_1(x) \in D_3$. 

Finally, let $a(x)+\omega b(x) \in T_3$. Then $a(x)+\omega b(x)=l(x)(g(x)+\omega k(x))+m(x)\omega h(x)$ for some $l(x),m(x)\in \F_p[x]$. Hence $a(x) \in D_1$ and $b(x)\in D_2$. This implies that $\wt(a(x)+\omega b(x))\geq \max\{d(D_1),d(D_2)\} $.  
\end{proof}

Note that if $D_i=0$ for any value $1\le i \le 4$, then we simply discard this code in the minimum distance lower bound of (\ref{distance}). For instance if $D_1=0$, then the minimum distance lower bound of (\ref{distance}) becomes 
$$\min\{d(D_3),d(D_4), d(D_2)\}\le d(C).$$
The following corollary gives a modification of this result to additive cyclic codes, which are generated by only one generator. In this result, the cyclic codes $C_i$ are obtained from $D_i$ after substituting $h(x)$ with $0$ in Theorem $\ref{distance bound 1}$ for $1\le i \le 3$. However, the code $C_4$ is obtained differently by considering a more direct observation. 
\begin{corollary}
Let $C=\langle g(x)+\omega k(x) \rangle _{\F_p[x]}$ be a length $n$ additive cyclic code over $\F_{p^2}$. Let $C_1$, $C_2$, $C_3$, and $C_4$ be the length $n$ linear cyclic codes over $\F_p$ generated by polynomials $g(x)$, $k(x)$, $\frac{x^n-1}{\gcd(x^n-1,g(x))} k(x)$, and $\frac{x^n-1}{\gcd(x^n-1,k(x))} g(x)$, respectively. Then 
\begin{equation}
\min\{d(C_3),d(C_4), \max \{d(C_1),d(C_2)\}\} \le d(C).
\end{equation}
\end{corollary}

\begin{proof}
As we mentioned above, the code $C_i$ all are obtained after applying the condition $h(x)=0$ in the structure of the codes $D_i$ for $1\le i\le 3$ in Theorem $\ref{distance bound 1}$. Since the code $D_4$ in Theorem $\ref{distance bound 1}$ is applied to bound the minimum weight of the set 
$$T_1=\{a(x) \in C: \ 0\neq a(x)\in \F_p[x] \},$$
we compute the minimum weight of $T_1$ directly in this proof. 
Let $f(x)\in T_1$. Then we can write it as $f(x)=a(x)(g(x)+\omega k(x))$ for some $a(x)\in \F_p[x]$. Hence $f(x)=a(x)g(x)$ and $a(x)k(x)\equiv 0 \pmod{x^n-1}$. This implies that $\frac{x^n-1}{\gcd(x^n-1,k(x))} \mid a(x)$. Hence $\frac{x^n-1}{\gcd(x^n-1,k(x))}g(x) \mid a(x)$ and we have $f(x) \in C_4$. 
\end{proof}

Next, we consider the restriction of the mentioned minimum distance bound to linear cyclic codes with the generator polynomials $g(x)+\omega k(x)$ and $h(x)$, where $k(x)=0$.

\begin{corollary}
Let $C=\langle g(x), \omega h(x) \rangle _{\F_p[x]}$ be a length $n$ additive cyclic code over $\F_{p^2}$. Let $E_1$ and $E_2$ be the length $n$ linear cyclic codes over $\F_p$ generated by polynomials $g(x)$ and $h(x)$, respectively. Then 
\begin{equation}
\min\{d(E_1),d(E_2) \} \le d(C).
\end{equation}
\end{corollary}

\begin{proof}
Applying the condition $k(x)=0$ to Theorem $\ref{distance bound 1}$ implies that $D_1=D_4=E_1$ and $D_2=D_3=E_2$. Now the result follows from the minimum distance bound of Theorem $\ref{distance bound 1}$.
\end{proof}
\section{Symplectic inner product and dual of additive cyclic codes}\label{S: quantum dual}

In this section, we determine generator polynomials of the symplectic dual of a given additive cyclic code over $\F_{p^2}$. Moreover, we give the generator polynomials of all self-orthogonal and self-dual codes. We also measure how close is a given additive cyclic code from being symplectic self-orthogonal. 
Recall that $p$ is a prime number and $n$ is a positive integer coprime to $p$. Moreover, elements of $\F_{p^2}$ are represented by $\mathbb{F}_{p^2}=\{\alpha+\beta \omega: \alpha,\beta \in \mathbb{F}_p\}$, 
where $\omega$ is a root of a degree 2 irreducible polynomial over $\mathbb{F}_p$. Recall that in (\ref{E: symp}) we defined the symplectic inner product of two elements in $\F_{p^2}^{n}$. 
We define the symplectic inner product of two polynomials analogously. In particular, for $c(x)=\ds\sum_{i=0}^{n-1}(a_i+\omega b_i)x^i$ and $c'(x)=\ds\sum_{i=0}^{n-1}(a'_i+\omega b'_i)x^i \in \mathbb{F}_{p^2}[x]/\langle x^n-1\rangle$, we define 
 $$c(x)*c'(x)=\sum_{i=0}^{n-1} (a_i b'_i-a'_ib_i).$$
 Here we use a different notation for the symplectic inner product to differentiate between the vectors and polynomials as different objects. 
 \begin{remark}\label{inner product of polynomials}
Let $c(x)=g_1(x)+\omega g_2(x)$ and $c'(x)=g_1'(x)+\omega g_2'(x)$ be two polynomials of $\mathbb{F}_{p^2}[x]/\langle x^n-1\rangle$, where $g_1(x),g_2(x),g_1'(x), g_2'(x) \in \F_p[x]/ \langle x^n-1\rangle$. 
Then $c(x)*c'(x)$ is the constant term of $g_1(x)g_2'(x^{-1})-g_2(x)g_1'(x^{-1}) \pmod{x^n-1}$. A similar argument shows that $c(x)*x^ic'(x)$ is the coefficient of $x^i$ in $g_1(x)g_2'(x^{-1})-g_2(x)g_1'(x^{-1}) \pmod{x^n-1}$. Thus if $g_1(x)g_2'(x^{-1})-g_2(x)g_1'(x^{-1})\equiv 0 \pmod{x^n-1}$, then the code generated by $c'(x)$ lies in the symplectic dual of the code generated by $c(x)$. We use this property very frequently through this section.
 \end{remark}

One can easily verify that the symplectic dual of an additive cyclic code $C$ over $\F_{p^2}$ is also an additive cyclic code over $\F_{p^2}$. Recall that by Theorem $\ref{decomposition}$ part (ii), each additive cyclic code of length $n$ over $\F_{p^2}$ can be represented uniquely as $C=\langle g_1(x)+\omega g_2(x),h(x)\rangle_{\F_p[x]}$, where $g_1(x),g_2(x),h(x)\in \F_p[x]/\langle x^n-1 \rangle$.
Our next theorem gives a criterion for the self-orthogonality of additive cyclic codes. The proof is very similar to that of \cite[Theorem 14 part c]{Calderbank}.

\begin{theorem}\label{criterion}
Let $C=\langle g_1(x)+\omega g_2(x),h(x)\rangle_{\F_p[x]} $ be a length $n$ additive cyclic code over $\F_{p^2}$. The code $C$ is self-orthogonal if and only if the following conditions are satisfied: 
\begin{enumerate}
\item $g_2(x)h(x^{-1})\equiv 0 \pmod{x^n-1}$,
\item $g_1(x)g_2(x^{-1})\equiv g_2(x)g_1(x^{-1}) \pmod{x^n-1}$.
\end{enumerate}
\end{theorem}

\begin{proof}
$\Rightarrow:$ Suppose that $C$ is self-orthogonal. For each $0\le i\le n-1$, the inner product of $g_1(x)+\omega g_2(x)$ and $x^ih(x)$ is the coefficient of $x^i$ in $-g_2(x)h(x^{-1}) \pmod{x^n-1}$. Since $C$ is self-orthogonal, we have $g_2(x)h(x^{-1})\equiv 0 \pmod{x^n-1}$. 
 Moreover, $\big(x^i(g_1(x)+\omega g_2(x))\big)\ast \big(g_1(x)+\omega g_2(x)\big)$ is the coefficient of $x^i$ in $g_1(x)g_2(x^{-1})-g_2(x)g_1(x^{-1})\pmod{x^n-1}$. Hence, for each $0\le i\le n-1$, the coefficient of $x^i$ in $g_1(x)g_2(x^{-1})-g_2(x)g_1(x^{-1})\pmod{x^n-1$} is zero. Thus $g_1(x)g_2(x^{-1})\equiv g_2(x)g_1(x^{-1}) \pmod{x^n-1}$.

$\Leftarrow:$ Conversely, the fact that $g_1(x)g_2(x^{-1})\equiv g_2(x)g_1(x^{-1}) \pmod{x^n-1}$ implies that all the vectors inside $\langle g_1(x)+\omega g_2(x)\rangle_{\F_p[x]}$ are self-orthogonal. Moreover, since $g_2(x)h(x^{-1})\equiv 0 \pmod{x^n-1}$, we conclude that $h(x)$ is orthogonal to all the cyclic shifts of $g_1(x)+\omega g_2(x)$. Finally $h(x)\ast x^ih(x)=0$ for each $0\le i \le n-1$. So $\langle g_1(x)+\omega g_2(x),h(x)\rangle_{\F_p[x]}$ is a symplectic self-orthogonal code.
\end{proof}

Recall that $x^n-1=\displaystyle \prod_{i=1}^{s} f_i(x)$, where each $f_i(x)$ is an irreducible polynomial in $\F_p[x]$. 
Moreover, as we mentioned earlier in $(\ref{equation1})$, we have $\mathbb{F}_{p^2}[x]/\langle x^n-1\rangle =\ds\bigoplus_{i=1}^s N_i$, where $N_i=\langle (x^n-1)/f_i(x)\rangle _{\F_{p^2}[x]}$. Let 
$\alpha$ be a primitive $n$-th root of unity in a finite filed extension of $\F_p$.
We denote the $p$-cyclotomic cosets modulo $n$ by $Z_i$ for each $1\le i \le s$ in the way that $f_i(x)=\ds\prod_{a \in Z_i}(x-\alpha^i)$. This gives a one-to-one correspondence between the sets $N_i$ and all the $p$-cyclotomic cosets modulo $n$. 
Our first goal in this section is to find the symplectic dual of a given additive cyclic code. In order to achieve this goal, we need a few preliminary results. In the next lemma, we find the symplectic dual of each $N_i$. 

\begin{lemma}\label{component dual}
Let $1\le i\le s$ and $C=N_i$. Then $ C^{\bot_s}=\ds\bigoplus_{\substack{k=1\\ k\neq j}}^{s} N_k$,
where $Z_j=-Z_i$.
\end{lemma}

\begin{proof}
First note that by Lemma \ref{case2} part (3) we have $C=\langle (x^n-1)/f_i(x), \omega ((x^n-1)/f_i(x)) \rangle _{\F_{p}[x]}$. 
If $Z_k\neq-Z_i$, then $f_i(x)\mid (x^n-1)/f_k(x^{-1})$ and $f_k(x)\mid (x^n-1)/f_i(x^{-1})$. So
we have 
\begin{itemize}
\item $\big((x^n-1)/f_i(x) \big)\big((x^n-1)/f_k(x^{-1}) \big)\equiv 0 \pmod{x^n-1}$ and
\item $\big((x^n-1)/f_k(x) \big)\big((x^n-1)/f_i(x^{-1}) \big)\equiv 0 \pmod{x^n-1}$.
\end{itemize}
 Hence the symplectic inner product of each element of $N_i$ and each element of $N_k$ is zero by definition. This proves that $\ds\bigoplus_{\substack{k=1\\ k\neq j}}^{s} N_k \subseteq C^{\bot_s}$. 
Note that both of $N_i$ and $N_j$ have $\F_p$-dimension $2\deg(f_i)$. Now, the facts that $\dim_{\F_p}(C)+\dim_{\F_p}(C^{\bot_s})=2n$ and $\dim_{\F_p}(C)=2\deg(f_i)$ implies the other inclusion. 
\end{proof}

Next, we find the symplectic dual of each irreducible additive cyclic code inside $N_i$ for $1\le i \le s$. 
\begin{lemma}\label{component dual 2}
Let $C\subsetneq N_i$ be a non-zero additive cyclic code for some $1\le i \le s$. Then 
\begin{equation}
 C^{\bot_s}=(\bigoplus_{\substack{k=1\\ k\neq j}}^{s} N_k) \bigoplus \langle g_1(x)+\omega g_2(x)\rangle _{\F_p[x]},
 \end{equation}
where $Z_j=-Z_i$ and
\begin{equation*}
g_1(x)+\omega g_2(x)=\begin{cases}
((x^n-1)/f_j(x))(s(x^{-1})+\omega) & \text{if}\ C=\langle\big((x^n-1)/f_i(x)\big)\big(\omega+s(x)\big) \rangle_{\F_p[x]} \\
(x^n-1)/f_j(x) & \text{if}\ C=\langle(x^n-1)/f_i(x) \rangle_{\F_p[x]} \\
\end{cases}.
\end{equation*}
\end{lemma}

\begin{proof}
By Lemma $\ref{component dual}$, one can see that $\ds\bigoplus_{\substack{k=1\\ k\neq j}}^{s} N_k\subseteq C^{\bot_s}$. Note that $\dim_{\F_p}( \langle g_1(x)+\omega g_2(x)\rangle _{\F_p[x]})=\dim_{\F_p}(C)$. So it is sufficient to show that $C$ is orthogonal to $g_1(x)+\omega g_2(x)$ and all its cyclic shifts. We prove the latter statement in two steps.
First suppose that $C=\langle\big((x^n-1)/f_i(x)\big)\big(\omega+s(x)\big) \rangle_{\F_p[x]}$ for some $s(x) \in \F_p[x]$. To show that the codes $C$ and $\langle g_1(x)+\omega g_2(x)\rangle_{\F_p[x]}$ are orthogonal, we apply Remark $\ref{inner product of polynomials}$. In particular, 
\begin{equation*}
\begin{split}
\big((x^n-1)/f_i(x)\big)s(x) &g_2(x^{-1})- \big((x^n-1)/f_i(x)\big)g_1(x^{-1})\equiv \big((x^n-1)/f_i(x)\big)s(x)\big((x^n-1)/f_i(x)\big)\\&-\big((x^n-1)/f_i(x)\big)\big((x^n-1)/f_i(x)\big)s(x)\equiv 0 \pmod{x^n-1}.
\end{split}
\end{equation*}
Next, suppose that $C=\langle (x^n-1)/f_i(x) \rangle_{\F_p[x]}$. Then
\begin{equation*}
\begin{split}
\big((x^n-1)/f_i(x)\big) g_2(x^{-1})- \big((x^n-1)/f_i(x)\big)g_1(x^{-1})&\equiv\big((x^n-1)/f_i(x)\big)0- 0\big((x^n-1)/f_j(x)\big)s(x)\\&\equiv0 \pmod{x^n-1}.
\end{split}
\end{equation*}
This shows that the code $C$ is orthogonal to the additive cyclic code generated by $g_1(x)+\omega g_2(x)$ and completes the proof.
\end{proof}

Note that as we showed in Lemma $\ref{component dual 2}$, when $C=\langle\big((x^n-1)/f_i(x)\big)\big(\omega+s(x)\big) \rangle_{\F_p[x]} $, its symplectic inner product contains the code $C'=\langle (x^n-1)/f_j(x))(s(x^{-1})+\omega)\rangle_{\F_p[x]}$. The code $C'$ is not in one of the forms given in Lemma $\ref{case2}$ part (4). In order to express the code $C'$ using the standard notation introduced in $\ref{case2}$ part (4), we choose its generator to be $g(x)=(x^n-1)/f_j(x))(t(x)+\omega)$, where $t(x)\equiv s(x^{-1}) \pmod{f_j(x)}$. Now it is easy to see that $g(x)$ belongs to the set $A$ introduced in Lemma $\ref{case2}$ part (4) and $C'=\langle (x^n-1)/f_j(x))(t(x)+\omega) \rangle_{\F_p[x]}$. 

Next, we combine the results of the previous two lemmas and the result of Theorem $\ref{decomposition}$ to determine generator polynomials of the symplectic dual for any additive cyclic code.

\begin{theorem}\label{dual form}
 Let $C$ be a length $n$ additive cyclic code over $\F_{p^2}$ such that $C=\displaystyle\bigoplus_{i=1}^{s} C_i$, where $C_i$ is an additive cyclic codes inside $N_i$ for each $1\le i\le s$.
Then $C^{\bot_s}=\langle \ds\sum_{i=1}^{s}g_i(x)+\omega k_i(x),\ds\sum_{i=1}^{s} \omega h_i(x) \rangle _{\F_p[x]}$, where for each $1\le i\le s$ we have $Z_j=-Z_i$ and
 \begin{itemize}
\item $g_i(x)=k_i(x)=h_i(x) =0$ if $C_j= N_j$,
\item $g_i(x)=h_i(x)=(x^n-1)/f_i(x)$ and $k_i(x)=0$ if $C_j=0$,
\item $g_i(x)+\omega k_i(x)= \big((x^n-1)/f_i(x)\big)\big(\omega+t_i(x)\big)$ and $h_i(x) =0$ if $C_j=\langle\big((x^n-1)/f_j(x)\big)\big(\omega+s_j(x)\big) \rangle_{\F_p[x]}$ and $t_i(x)\equiv s_j(x^{-1}) \pmod{f_j(x)}$,
\item $g_i(x)=(x^n-1)/f_i(x)$ and $ k_i(x)=h_i(x) =0$ if $C_j=\langle(x^n-1)/f_j(x) \rangle_{\F_p[x]}$.

\end{itemize} 
\end{theorem}

\begin{proof}
We apply Lemmas $\ref{component dual}$ and $ \ref{component dual 2}$ to prove the statement. If $C_j= N_j$, then $C^{\bot_s}\cap N_i=\{0\}$ by Lemma $\ref{component dual}$. Moreover, by the same lemma, if $C_j=0$, then $N_i \subseteq C^{\bot_s}$. This proves the first two bullets. 
Finally, Lemma $ \ref{component dual 2}$ implies that 
\begin{itemize}
\item $\langle \big((x^n-1)/f_i(x)\big)\big(\omega+t_i(x)\big) \rangle_{\F_p[x]}\subseteq C^{\bot_s}$ if $C_j=\langle\big((x^n-1)/f_j(x)\big)\big(\omega+s_j(x)\big) \rangle_{\F_p[x]}$, and
\item $\langle (x^n-1)/f_i(x) \rangle_{\F_p[x]}\subseteq C^{\bot_s}$ if $C_j=\langle (x^n-1)/f_i(x) \rangle_{\F_p[x]}$.
\end{itemize}
This proves the statements of the last two bullets.
\end{proof}

To determine self-orthogonal and self-dual additive cyclic codes over $\F_{p^2}$, we need more information about irreducible factors of $x^n-1$ over $\F_p$. 
Let $Z_1,Z_2,\ldots,Z_{r}$ and $Z_1',-Z_1',\ldots,Z_t',-Z_t'$ be all the $p$-cyclotomic cosets modulo $n$, where $Z_i=-Z_i$ and $r+2t=s$. 
Each $Z_i$ is in correspondence to an irreducible polynomial $f_i(x)$ and $(Z_j',-Z_j')$ are in correspondence with an irreducible pair of polynomials $(f_{j1}(x), f_{j2}(x))$ over $\F_p$. Therefore, we can rewrite the irreducible decomposition of $x^n-1$ as 
$$x^n-1=\ds\prod_{i=1}^{r}f_i(x) \prod_{j=1}^{t}f_{j1}(x)f_{j2}(x).$$
 We use the above representation of cyclotomic cosets in the upcoming results. Next, we classify self-orthogonal and self-dual additive codes over $\F_{p^2}$.

\begin{theorem}\label{self-orthogonal}
Let $C$ be a length $n$ additive cyclic code over $\F_{p^2}$ such that $C=\displaystyle\bigoplus_{i=1}^{s} C_i$, where $C_i$ is an additive cyclic codes inside $N_i$ for each $1\le i\le s$. Then $C$ is symplectic self-orthogonal if and only if
\begin{enumerate}
\item for all $1\le k\le r$ only one of the following holds.
\begin{enumerate}
\item $C_k=0$.
\item$C_k= \langle ((x^n-1)/f_k(x))(s(x)+\omega) \rangle_{\F_p[x]}$, where $f_k \mid s(x^{-1})-s(x)$.
\item$C_k=\langle (x^n-1)/f_k(x)\rangle_{\F_p[x]}$.
\end{enumerate}
\item for all $1 \le j \le t$ only one of the following holds. 
\begin{enumerate}
\item $C_{j1}=0$ or $C_{j2}=0$.
\item $C_{j1}= \langle ((x^n-1)/f_{j1}(x))(s(x)+\omega) \rangle_{\F_p[x]}$ and $C_{j2}= \langle ((x^n-1)/f_{j2}(x))(s(x^{-1})+\omega) \rangle_{\F_p[x]}$.
\item $C_{j1}= \langle (x^n-1)/f_{j1}(x) \rangle_{\F_p[x]}$ and $C_{j2}= \langle (x^n-1)/f_{j2}(x)\rangle_{\F_p[x]}$.
\end{enumerate}
\end{enumerate}
\end{theorem}

\begin{proof}
First, let $1\le k\le r$. 
By Lemma $\ref{component dual}$, if $C_k=N_k$, then $C^{\bot_s} \cap N_k=\{0\}$. So $C_k$ cannot have two generator polynomials. 
Moreover, by Lemma $\ref{component dual 2}$, if $0 \neq C_k= \langle ((x^n-1)/f_k(x))(s(x)+\omega) \rangle_{\F_p[x]}$, then $ \langle ((x^n-1)/f_k(x))(s(x^{-1})+\omega) \rangle_{\F_p[x]}\subseteq C^{\bot_s}$. Thus $C_k$ is self-orthogonal if and only if
 $$C_k=\langle ((x^n-1)/f_k(x))(s(x)+\omega) \rangle_{\F_p[x]} = \langle ((x^n-1)/f_k(x))(s(x^{-1})+\omega) \rangle_{\F_p[x]}.$$
Note that the above equality holds if and only if $f_k \mid s(x^{-1})-s(x)$. 
Thus $C_k$ is self-orthogonal if and only if one of the conditions of Part (1) follows.  

Next, let $1 \le j \le t$. By Lemma $\ref{component dual}$, if $C_{j1}=N_{j1}$, then $C^{\bot_s} \cap N_{j2}=\{0\}$. So if one of $C_{j1}$ or $C_{j2}$ has two generator polynomials, the other code should be zero. Moreover, the same lemma shows that if $C_{j1}=0$ or $C_{j2}=0$, then $C_{j1}+C_{j2}$ is self-orthogonal.
So we concentrate only on the case when both $C_{j1}$ and $C_{j2}$ have exactly one non-zero generator.
By Lemma $\ref{component dual 2}$, if $C_{j1}= \langle ((x^n-1)/f_{j1}(x))(s(x)+\omega) \rangle_{\F_p[x]}$, then $$C^{\bot_s} \cap N_{j2}=\langle ((x^n-1)/f_{j2}(x))(s(x^{-1})+\omega) \rangle_{\F_p[x]}.$$
In this case, the code $C_{j1}\oplus C_{j2}$ is self-orthogonal if and only if condition (2)(b) is satisfied. Condition (2)(c) follows similarly by applying Lemma $\ref{component dual 2}$.
\end{proof}

Next, we use the above conditions to characterize all the symplectic self-dual additive cyclic codes over $\F_{p^2}$.

\begin{corollary}
Let $C$ be a length $n$ additive cyclic code over $\F_{p^2}$ such that $C=\displaystyle\bigoplus_{i=1}^{s} C_i$, where $C_i$ is an additive cyclic codes inside $N_i$ for each $1\le i\le s$. Then $C$ is symplectic self-dual if and only if
\begin{enumerate}
\item for all $1\le k\le r$ only one of the following holds.
\begin{enumerate}
\item$C_k= \langle ((x^n-1)/f_k(x))(s(x)+\omega) \rangle_{\F_p[x]}$ where $f_k \mid s(x^{-1})-s(x)$.
\item$C_k=\langle (x^n-1)/f_k(x)\rangle_{\F_p[x]}$.
\end{enumerate}
\item for all $1 \le j \le t$ only one of the following holds.
\begin{enumerate}
\item $C_{j1}=0$ and $C_{j2}=N_{j2}$.
\item $C_{j2}=0$ and $C_{j1}=N_{j1}$.
\item $C_{j1}= \langle ((x^n-1)/f_{j1}(x))(s(x)+\omega) \rangle_{\F_p[x]}$ and $C_{j2}= \langle ((x^n-1)/f_{j2}(x))(s(x^{-1})+\omega) \rangle_{\F_p[x]}$.
\item $C_{j1}= \langle (x^n-1)/f_{j1}(x) \rangle_{\F_p[x]}$ and $C_{j2}= \langle (x^n-1)/f_{j2}(x)\rangle_{\F_p[x]}$.
\end{enumerate}
\end{enumerate}
\end{corollary}

\begin{proof}
Note that all the self-dual additive cyclic codes over $\F_{p^2}$ satisfy the conditions of Theorem $\ref{self-orthogonal}$ and have maximal dimension. Thus the result easily follows by implying the maximal property into the conditions of theorem $\ref{self-orthogonal}$.
\end{proof}

Our next goal is to compute the parameter $e=\dim_{\F_p}(C)-\dim_{\F_p}(C\cap C^{\bot_s})$ for all additive cyclic codes. The parameter $e $ determines how close an additive cyclic code $C$ is from being self-orthogonal. This parameter plays an important role in the quantum construction that we are applying in the next section. 

\begin{theorem}\label{e parameter}
Let $C$ be a length $n$ additive cyclic code over $\F_{p^2}$ such that $C=\displaystyle\bigoplus_{i=1}^{s} C_i$, where $C_i$ is an additive cyclic codes inside $N_i$ for each $1\le i\le s$. Let
\begin{enumerate}
\item $B_1=\{\alpha_1,\alpha_2,\ldots,\alpha_{t_1}\} \subseteq \{1,2,\ldots,r\}$ such that $C_{\alpha_l}=N_{\alpha_l}$ for all $1\le l \le t_1$,
\item $B_2=\{\beta_1,\beta_2,\ldots,\beta_{t_2}\} \subseteq \{1,2,\ldots,r\}$ such that $C_{\beta_l}= \langle ((x^n-1)/f_{\beta_l}(x))(s_{\beta_l}(x)+\omega) \rangle _{\F_{p}[x]}$ and $f_{\beta_l} \nmid s_{\beta_l}(x^{-1})-s_{\beta_l}(x)$ for all $1\le l \le t_2$,

\item $B_3=\{\gamma_1,\gamma_2,\ldots,\gamma_{t_3}\} \subseteq \{1,2,\ldots,t\}$ such that one of $C_{\gamma_l1}$ and $C_{\gamma_l2}$ is generated by two polynomials and the other one has only one generator polynomial for all $1\le l \le t_3$,

\item $B_4=\{\kappa_1,\kappa_2,\ldots,\kappa_{t_4}\} \subseteq \{1,2,\ldots,t\}$ such that both of $C_{\kappa_l1}$ and $C_{\kappa_l2}$ are generated by two polynomials for all $1\le l \le t_4$,
\item $B_5=\{\sigma_1,\sigma_2,\ldots,\sigma_{t_5}\} \subseteq \{1,2,\ldots,t\}$ such that both of $C_{\sigma_l1}$ and $C_{\sigma_l2}$ are generated by one polynomial for all $1\le l \le t_5$ and 
\begin{enumerate}
\item if $C_{\sigma_l1}= \langle ((x^n-1)/f_{\sigma_l1}(x))(s_{\sigma_l}(x)+\omega) \rangle_{\F_p[x]}$, then $C_{\sigma_l2}\neq \langle ((x^n-1)/f_{\sigma_l2}(x))(s_{\sigma_l}(x^{-1})+\omega) \rangle_{\F_p[x]}$.
\item if $C_{\sigma_l1}= \langle (x^n-1)/f_{\sigma_l1}(x) \rangle_{\F_p[x]}$, then $C_{\sigma_l2}\neq \langle (x^n-1)/f_{\sigma_l2}(x)) \rangle_{\F_p[x]}$.
\end{enumerate}
\end{enumerate}
Then
\begin{equation}\label{e computation}
e=\dim_{\F_p}(C)-\dim_{\F_p}(C\cap C^{\bot_s})=\sum_{l=1}^{t_1}2 |Z_{\alpha_l}| +\sum_{l=1}^{t_2} |Z_{\beta_l}| + \sum_{l=1}^{t_3}2|Z_{\gamma_l}| +\sum_{l=1}^{t_4}4|Z_{\kappa_l}|+ \sum_{l=1}^{t_5}2|Z_{\sigma_l}|.
\end{equation}
\end{theorem}

\begin{proof}
By Theorem $\ref{self-orthogonal}$, an additive cyclic code is not symplectic self-orthogonal if and only if at least one of the sets $B_1-B_5$ is non-empty. Next, we consider all scenarios (1)-(5) independently.

\begin{enumerate}
\item Let $j\in B_1$. In this case, $C^{\bot_s} \cap C_j=\{0\}$ which implies that $\dim_{\F_p}(C_j)-\dim_{\F_p}(C_j\cap C^{\bot_s})=2|Z_j|$.

\item Let $j\in B_2$. In this case, $C^{\bot_s} \cap C_j=\{0\}$ which implies that $\dim_{\F_p}(C_j)-\dim_{\F_p}(C_j\cap C^{\bot_s})=|Z_j|$.

\item Let $j\in B_3$. Without loss of generality we assume that $C_{j1}=N_{j1}$ and $C_{j2}$ is an irreducible subcode of $N_{j2}$. 
In this case, the intersection $C^{\bot_s} \cap (C_{j1}\oplus C_{j2})$ is an irreducible subcode of $N_{j1}$ which implies that $\dim_{\F_p}(C_j)-\dim_{\F_p}((C_{j1}\oplus C_{j2})\cap C^{\bot_s})=3|Z_j|-|Z_j|=2|Z_j|$.

\item Let $j\in B_4$. 
In this case, $C^{\bot_s} \cap (C_{j1}\oplus C_{j2})=\{0\}$ which implies that $\dim_{\F_p}(C_j)-\dim_{\F_p}((C_{j1}\oplus C_{j2})\cap C^{\bot_s})=4|Z_j|$.

\item Let $j\in B_5$. In both parts (a) and (b), $C^{\bot_s} \cap (C_{j1}\oplus C_{j2})=\{0\}$ which implies that $\dim_{\F_p}(C_j)-\dim_{\F_p}((C_{j1}\oplus C_{j2})\cap C^{\bot_s})=2|Z_j|$.

\end{enumerate}
Now, the result follows by combining the above observations. 

\end{proof}

Note that the case (2) of Theorem \ref{e parameter} never happens for $C_i$ with $\deg(f_i(x))=1$. Moreover,
for each $1\le i \le r$, the cyclotomic coset $Z_i$ either is a singleton or it has an even size. This is mainly because for each $0\neq a \in Z_i$, if $a \equiv -a \pmod n$, then $n \mid 2a$, which implies that $n$ is even. Hence in this case $p\neq 2$ (we assumed that $\gcd(n,p)=1$) and $Z_i=\{a\}$. 
Therefore, if $Z_i$ satisfies the case (2) of Theorem \ref{e parameter} and $|Z_i|>1$, then for any $a\in Z_i$, we have $-a\in Z_i$ and $-a \not\equiv a \pmod n$. This implies that $|Z_i|$ is an even integer. 
This fact and the formula in $(\ref{e computation})$ imply that the nearly self-orthogonality parameter $e$ of an additive cyclic code is always an even integer. Next, we classify additive cyclic codes with small values of $e$. First, we need the following preliminary result.

\begin{lemma}
Let $p$ be a prime number and $\gcd(n,p)=1$ for some positive number $n$.
\begin{enumerate}
\item[(i)] If $\gcd(n,p-1)=d$, then there are $d$ singleton $p$-cyclotomic cosets modulo $n$ and all of their coset leaders are $\{k\frac{n}{d} :0\le k \le d-1\}$.

\item[(ii)] If $\gcd(n,p-1)=d$ and $\gcd(n,p^2-1)=d'$. Then there are $\frac{d'-d}{2}$ $p$-cyclotomic cosets modulo $n$ of size two. 
\end{enumerate}
\end{lemma}

\begin{proof}
(i) The proof easily follows from the fact that $\{a\}$ is a singleton coset if and only if $a\equiv pa \pmod{n}$ or equivalently if and only if $a(p-1)\equiv 0\pmod{n}$. By elementary number theory, if $\gcd(n,p-1)=d$, then the latter equation has $d$ solutions in the forms $\{k\frac{n}{d} :0\le k \le d-1\}$.

(ii) A $p$-cyclotomic coset modulo $n$ containing $a$ has size two if and only if $a\equiv p^2a \pmod{n}$ and $a\not\equiv pa \pmod{n}$. So we get $d'$ candidate for the size two cosets by solving $a\equiv p^2a \pmod{n}$. Moreover, each singleton cyclotomic coset is formed by a solution of the latter equation. Note also that the $p$-cyclotomic coset of size two containing $a$ and $pa$ is counted twice in our previous observation. Hence there are
$\frac{d'-d}{2}$ many different cosets. 
\end{proof}

For example, for an odd $n$, the only singleton $p$-cyclotomic coset modulo $n$ is $\{0\}$ when $p=2$ or $p=3$. If $n$ is even, then $\{\frac{n}{2}\}$ and $\{0\}$ are the only singleton cyclotomic cosets for $p=3$. The next theorem classifies all the additive cyclic codes with $e=2$. Note that the case $e=0$ happens if an additive cyclic code is symplectic self-orthogonal, and this case was characterized in Theorem $\ref{self-orthogonal}$.

\begin{theorem}
Let $C=\displaystyle\bigoplus_{i=1}^{s} C_i$ be an additive cyclic code of length $n$ over $\F_{p^2}$. Then 
$$e=\dim_{\F_p}(C)-\dim_{\F_p}(C\cap C^{\bot_s})=2$$
 if and only if all $C_i$ satisfy the conditions of Theorem $\ref{self-orthogonal}$ except one which is in correspondence to
 \begin{enumerate}
 \item a singleton coset and satisfies condition $(1)$ of Theorem $\ref{e parameter}$,
 \item a size two coset and satisfies condition $(2)$ of Theorem $\ref{e parameter}$.
 \end{enumerate}
\end{theorem}

\begin{proof}
The result follows from considering the formula $(\ref{e computation})$ and considering all conditions of Theorem $\ref{e parameter}$.  
\end{proof}

Many of our record-breaking quantum codes provided in the next section have $e=2$. In general, the total number of all additive cyclic codes can be a very large number. So the classification of $e$ values significantly helps to prune the search algorithm for quantum codes with good parameters. 

\section{New binary quantum codes}\label{S: new codes}

In this section, we first recall a construction of binary quantum codes from additive codes, which does not require the symplectic self-orthogonality condition of Theorem \ref{T: quantum def}. Then we apply this construction to several nearly self-orthogonal additive cyclic codes over $\F_4$ and construct new binary quantum codes. In the rest of this section, we show the quaternary filed by $\F_4=\{0,1,\omega, \omega+1\}$, where $\omega^2=\omega+1$.

\begin{theorem}\cite[Corollary 3.3.7]{Reza},\cite{Reza2}\label{construction}
Let $C$ be an $(n,2^k)$ additive code over $\F_4$ and
$$r= \frac{2n-k - \dim_{\F_p}(C \cap C^{\bot_s})}{2}$$ Then there exists a binary quantum code with parameters $[[n+r,k-n+r,d]]_2$, where $$d\geq\min\{d(C),d(C+C^{\bot_s})+1\}.$$ 
\end{theorem}

Note that we take advantage of the result of Theorem \ref{e parameter} in the computation of Theorem \ref{construction}. In particular, the value of $r$ in Theorem \ref{construction} is $\frac{\dim_{\F_p}(C^{\bot_s})-\dim_{\F_p}(C \cap C^{\bot_s})}{2}$, where the numerator measures the nearly self-orthogonality of the code $C^{\bot_s}$.
 Next, we briefly describe two of our new binary quantum codes. The rest of our new binary quantum codes presented in Table \ref{Ta: 1} can be constructed in a similar way.

\begin{example}
Let $n=21$ and $C=\langle g(x)+\omega k(x) \rangle_{\F_2[x]}$ be an additive cyclic code over $\F_4$, where $$g(x)=x^{20} + x^{17} + x^{15} + x^{13} + x^{11} + x^8 + x^7 + x^6 + x^5 + x^4 + x^3 + 1$$ and $$k(x)=x^{19} + x^{18} + x^{17} + x^{16} + x^{14} + x^{10} + x^5 + x^4 + x^3 + x^2 + x + 1.$$
The code $C$ is a $(21,2^{20})$ additive code. Moreover, our computation using the result of Theorem \ref{e parameter} shows that $C$ has nearly self-orthogonality parameter $e=2$. 
Moreover, 
$$7=\min\{d(C^{\bot_s}),d(C+C^{\bot_s})+1\}.$$
So, applying the construction of Theorem \ref{construction} to the code $C^{\bot_s}$ gives a {\em new quantum code} with parameters $[[22,2,7]]_2$. It has a better minimum distance than the previous best-known quantum code with the same length and dimension, which had minimum distance $6$.
\end{example}

\begin{example}
Let $n=35$ and $C=\langle g(x)+\omega k(x) \rangle_{\F_2[x]}$ be an additive cyclic code over $\F_4$, where $$g(x)=x^{33} + x^{29} + x^{28} + x^{24} + x^{19} + x^{18} + x^{15} + x^{13} + x^{12} + x^{11} + x^6 + x^4
  + x + 1$$ and $$k(x)=x^{34} + x^{33} + x^{31} + x^{30} + x^{29} + x^{27} + x^{25} + x^{23} + x^{22} + x^{20} + x^{19} +x^{18} + x^{15} + x^{12} + x^8 + x^3 + x.$$
The code $C$ has parameters $(35,2^{20})$ as an additive cyclic code over $\F_{4}$. Also, the result of Theorem \ref{e parameter} shows that $C$ has nearly self-orthogonality parameter $e=4$. Moreover, 
$$6=\min\{d(C^{\bot_s}),d(C+C^{\bot_s})+1\}.$$
So, applying the construction of Theorem \ref{construction} to the code $C^{\bot_s}$ gives a {\em record-breaking quantum code} with parameters $[[37,17,6]]_2$. The previous best-known binary quantum code with the same parameters had minimum distance $5$.
\end{example}

In general, in order to apply the quantum construction given in Theorem $\ref{construction}$, we target additive cyclic codes with the nearly self-orthogonality $e\le 4$. Because it is more likely to get a new quantum code when $e$ value is small.
In Table \ref{Ta: 1}, we present the parameters of our new binary quantum codes. In the table, we start with an additive cyclic code $C$ over $\F_{4}$ and compute its nearly self-orthogonality. Then we apply the quantum construction of Theorem \ref{construction} to the code $C^{\bot_s}$.
The parameters of the corresponding quantum code are given in the fourth column. Moreover, the minimum distance of the previous quantum code with the same length and dimension is provided in the last column of the table. The previous minimum distance is taken from Grassl's code table \cite{GrasslT}. We record the generator polynomials of the additive cyclic codes of Table \ref{Ta: 1} in Table \ref{T: 2}.

\begin{table}[h]
\begin{center}
\begin{tabular}{|p{0.5 cm} |p{1.4 cm}|p{1.5 cm}|p{2.6 cm}| p{3.1 cm}|}
\hline
 No & Length & e value &Parameters &Previous distance\\
 \hline
 1& $n=21$ & $ 2$ & $[[22,2,\textbf{7}]]_2$& 6\\
 2& $n=35$ & $ 4$ & $[[37,17,\textbf{6}]]_2$& 5\\
 3& $n=45$ & $ 0$ & $[[45,6,\textbf{10}]]_2$& 9\\
 4& $n=45$ & $ 0$ & $[[45,10,\textbf{9}]]_2$& 8\\
 5& $n=51$ & $ 0$ & $[[51,8,\textbf{11}]]_2$& 10\\
 6& $n=51$ & $ 2$ & $[[52,16,\textbf{10}]]_2$& 9\\
 7& $n=51$ & $ 2$ & $[[52,24,\textbf{8}]]_2$& 7\\
 8& $n=63$ & $ 2$ & $[[64,33,\textbf{8}]]_2$& 7\\
 9& $n=63$ & $ 2$ & $[[64,34,\textbf{8}]]_2$& 7\\
 10& $n=63$ & $ 2$ & $[[64,35,\textbf{8}]]_2$& 7\\

 \hline
\end{tabular}
\vspace{2mm}
\end{center}
\caption{Parameters of new binary quantum codes.} 
\label{Ta: 1}
\end{table}

Note also that applying the secondary constructions presented in Theorem \ref{T: secondary} to the new codes of Table \ref{Ta: 1} produces many more record-breaking quantum codes. In particular, the new $[[52,24,8]]_2$ quantum codes alone produces the following new quantum codes: 
$$[[52,21,8]]_2, [[52,22,8]]_2,[[52,23,8]]_2,[[53,21,8]]_2,[[53,22,8]]_2,[[53,23,8]]_2,[[53,24,8]]_2.$$
Around the same time as us, authors of \cite{Guan} independently found several new binary quantum codes by applying the connection between quasi-cyclic codes and additive cyclic codes. In particular, three of our new quantum codes, namely $[[45,6,10]], [[45,45,10,9]]$, and $[[51,8,11]]$, are also among the new quantum codes of \cite{Guan}.

\section*{Acknowledgement}
The authors would like to thank Petr Lison\v{e}k and Markus Grassl for many
interesting discussions and comments.

\begin{table}[ht]
\begin{center}
\scalebox{0.75}{
\begin{tabular}{|p{0.03\textwidth}|p{0.97\textwidth}|}
  \hline
   No & Generator polynomials as in Theorem \ref{decomposition} part (II)\\
  \hline
  \multirow{3}{*}1&
  $\textbf{g(x)}=x^{20} + x^{17} + x^{15} + x^{13} + x^{11} + x^8 + x^7 + x^6 + x^5 + x^4 + x^3 + 1
$\\
  & $\textbf{k(x)}=x^{19} + x^{18} + x^{17} + x^{16} + x^{14} + x^{10} + x^5 + x^4 + x^3 + x^2 + x + 1
$\\
  &$\textbf{h(x)}=0$\\
  \hline
   \multirow{3}{*}2&
  $\textbf{g(x)}=x^{33} + x^{29} + x^{28} + x^{24} + x^{19} + x^{18} + x^{15} + x^{13} + x^{12} + x^{11} + x^6 + x^4 
  + x + 1$\\
  &$\textbf{k(x)}=x^{34} + x^{33} + x^{31} + x^{30} + x^{29} + x^{27} + x^{25} + x^{23} + x^{22} + x^{20} + x^{19} + 
  x^{18} + x^{15} + x^{12} + x^8 + x^3 + x$\\
  &\textbf{h(x)}=0\\
  \hline
   \multirow{3}{*}3&
  $\textbf{g(x)}=x^{44 }+ x^{43 }+ x^{41 }+ x^{40 }+ x^{39 }+ x^{38 }+ x^{34 }+ x^{33 }+ x^{30 }+ x^{26 }+ x^{24 }+
  x^{20 }+ x^{19 }+ x^{18 }+ x^{17 }+ x^{16 }+ x^{15 }+ x^{14 }+ x^{11 }+ x^{9 }+ x^{5 }+ x^{3 }+ 1$\\
  &$\textbf{k(x)}=x^{43 }+ x^{42 }+ x^{41 }+ x^{40 }+ x^{36 }+ x^{33 }+ x^{32 }+ x^{31 }+ x^{30 }+ x^{28 }+ x^{26 }+
  x^{25 }+ x^{17 }+ x^{16 }+ x^{15 }+ x^{13 }+ x^{11 }+ x^{10 }+ x^{2 }+ x$\\
  &$\textbf{h(x)}=0$\\
  \hline
   \multirow{3}{*}4&
  $\textbf{g(x)}=x^{44 }+ x^{43 }+ x^{40 }+ x^{38 }+ x^{37 }+ x^{34 }+ x^{31 }+ x^{27 }+ x^{22 }+ x^{21 }+ x^{20 }+
  x^{19 }+ x^{18 }+ x^{17 }+ x^{14 }+ x^{12 }+ x^{7 }+ x^{6 }+ x^{5 }+ x^{3 }+ x + 1$\\
  &$\textbf{k(x)}=x^{44 }+ x^{41 }+ x^{40 }+ x^{37 }+ x^{36 }+ x^{35 }+ x^{33 }+ x^{30 }+ x^{29 }+ x^{27 }+ x^{26 }+
  x^{25 }+ x^{22 }+ x^{20 }+ x^{15 }+ x^{14 }+ x^{12 }+ x^{11 }+ x^{10 }+ x^{7 }+ x^{5}$\\
  &$\textbf{h(x)}=0$\\
  \hline
       \multirow{3}{*}5&
  $\textbf{g(x)}=x^{50 }+ x^{49 }+ x^{48 }+ x^{46 }+ x^{45 }+ x^{43 }+ x^{42 }+ x^{41 }+ x^{40 }+ x^{37 }+ x^{36 }+
  x^{35 }+ x^{30 }+ x^{29 }+ x^{28 }+ x^{26 }+ x^{23 }+ x^{19 }+ x^{18 }+ x^{17 }+ x^{16 }+ x^{15 }+
  x^{14 }+ x^{13 }+ x^{9 }+ x^{7 }+ x^{6 }+ x
$\\
  &$\textbf{k(x)}=x^{50 }+ x^{47 }+ x^{44 }+ x^{43 }+ x^{42 }+ x^{41 }+ x^{40 }+ x^{38 }+ x^{36 }+ x^{35 }+ x^{33 }+
  x^{32 }+ x^{28 }+ x^{26 }+ x^{24 }+ x^{21 }+ x^{20 }+ x^{16 }+ x^{14 }+ x^{12 }+ x^{9 }+ x^{8 }+
  x^{7 }+ x + 1$\\
  &\textbf{h(x)}=0\\
  \hline
   \multirow{3}{*}6&
  $\textbf{g(x)}=x^{48} + x^{40} + x^{37} + x^{36} + x^{33} + x^{31} + x^{30} + x^{24} + x^{23} + x^{21} + x^{19} +
  x^{15} + x^{11} + x^{10} + x^9 + x^8 + x^7 + x^4 + x^3 + x + 1$\\
  &$\textbf{k(x)}=x^{41} + x^{40} + x^{36} + x^{35} + x^{34} + x^{33} + x^{30} + x^{29} + x^{27} + x^{23} + x^{22} +
  x^{21} + x^{19} + x^{18} + x^{16} + x^{13} + x^{12} + x^{10} + x^9 + x^8 + x^7 + x^6 + x^5
  + x^4 + x^3 + x$\\
  &$\textbf{h(x)}=x^{50} + x^{49} + x^{48} + x^{47} + x^{46} + x^{45} + x^{44} + x^{41} + x^{40} + x^{39} + x^{33} +
  x^{31} + x^{30} + x^{28} + x^{25} + x^{22} + x^{20} + x^{19} + x^{18} + x^{17} + x^{16} + x^{14} +
  x^{13} + x^{11} + x^9 + x^5 + x^4$\\
  \hline
   \multirow{3}{*}7&
  $\textbf{g(x)}=x^{49} + x^{48} + x^{46} + x^{44} + x^{43} + x^{41} + x^{38} + x^{37} + x^{36} + x^{33} + x^{32} +
  x^{31} + x^{30} + x^{29} + x^{27} + x^{25} + x^{21} + x^{20} + x^{18} + x^{17} + x^{15} + x^{11} +
  x^{10} + x^7 + x^2$\\
  &$\textbf{k(x)}=x^{43} + x^{42 }+ x^{41 }+ x^{40} + x^{38} + x^{37} + x^{33} + x^{32} + x^{30} + x^{26} + x^{24} +
  x^{22} + x^{19} + x^{18} + x^{16} + x^{15} + x^{13} + x^9 + x^5 + x^4 + x^2 + 1$\\
  &$\textbf{h(x)}=x^{50} + x^{49} + x^{48} + x^{47} + x^{46} + x^{45} + x^{44} + x^{43} + x^{42} + x^{41} + x^{40} +
  x^{39} + x^{38} + x^{37} + x^{36} + x^{35} + x^{34} + x^{33} + x^{32} + x^{31} + x^{30} + x^{29} +
  x^{28} + x^{27} + x^{26} + x^{25} + x^{24} + x^{23} + x^{22} + x^{21} + x^{20} + x^{19} + x^{18} +
  x^{17} + x^{16} + x^{15} + x^{14} + x^{13} + x^{12} + x^{11} + x^{10} + x^9 + x^8 + x^7 +
  x^6 + x^5 + x^4 + x^3 + x^2 + x + 1$\\
  \hline
   \multirow{3}{*}8&
  $\textbf{g(x)}=x^{61 }+ x^{60 }+ x^{59 }+ x^{57 }+ x^{56 }+ x^{53 }+ x^{52 }+ x^{51 }+ x^{42 }+ x^{41 }+ x^{38 }+
  x^{36 }+ x^{34 }+ x^{32 }+ x^{31 }+ x^{28 }+ x^{27 }+ x^{26 }+ x^{24 }+ x^{20 }+ x^{19 }+ x^{16 }+
  x^{14 }+ x^{13 }+ x^{12 }+ x^{11 }+ x^{9 }+ x^{8 }+ x^{7 }+ x^{6 }+ x^{5 }+ x^{4 }+ x^{3 }+ x^{2 }+
  x$\\
  &$\textbf{k(x)}=x^{61 }+ x^{59 }+ x^{57 }+ x^{56 }+ x^{55 }+ x^{54 }+ x^{52 }+ x^{51 }+ x^{50 }+ x^{49 }+ x^{47 }+
  x^{44 }+ x^{37 }+ x^{36 }+ x^{35 }+ x^{33 }+ x^{32 }+ x^{31 }+ x^{29 }+ x^{28 }+ x^{27 }+ x^{26 }+
  x^{24 }+ x^{22 }+ x^{21 }+ x^{16 }+ x^{8 }+ x^{5 }+ x^{3 }+ x^{2}$\\
  &$\textbf{h(x)}=x^{62 }+ x^{61 }+ x^{60 }+ x^{59 }+ x^{58 }+ x^{51 }+ x^{49 }+ x^{47 }+ x^{44 }+ x^{43 }+ x^{40 }+
  x^{36 }+ x^{33 }+ x^{31 }+ x^{30 }+ x^{28 }+ x^{27 }+ x^{23 }+ x^{22 }+ x^{21 }+ x^{19 }+ x^{14 }+
  x^{13 }+ x^{11 }+ x^{9 }+ x^{8 }+ x^{7 }+ x^{4 }+ x^{3 }+ x^{2 }+ x$\\
  \hline
   \multirow{3}{*}{9}&
  $\textbf{g(x)}=x^{60 }+ x^{59 }+ x^{58 }+ x^{55 }+ x^{54 }+ x^{53 }+ x^{52 }+ x^{51 }+ x^{48 }+ x^{47 }+ x^{45 }+
  x^{44 }+ x^{40 }+ x^{38 }+ x^{37 }+ x^{36 }+ x^{35 }+ x^{34 }+ x^{33 }+ x^{32 }+ x^{31 }+ x^{30 }+
  x^{29 }+ x^{28 }+ x^{27 }+ x^{24 }+ x^{23 }+ x^{22 }+ x^{21 }+ x^{15 }+ x^{13 }+ x^{10 }+ x^{9 }+
  x^{7 }+ x^{6 }+ x^{3 }+ x + 1$\\
  &$\textbf{k(x)}=x^{62 }+ x^{59 }+ x^{56 }+ x^{55 }+ x^{54 }+ x^{53 }+ x^{49 }+ x^{47 }+ x^{46 }+ x^{42 }+ x^{41 }+
  x^{40 }+ x^{37 }+ x^{35 }+ x^{33 }+ x^{31 }+ x^{29 }+ x^{28 }+ x^{27 }+ x^{24 }+ x^{20 }+ x^{19 }+
  x^{16 }+ x^{15 }+ x^{14 }+ x^{7 }+ x^{4 }+ x^{2 }+ x$\\
  &$\textbf{h(x)}=0$\\
  \hline
   \multirow{3}{*}{10}&
  $\textbf{g(x)}=x^{61 }+ x^{60 }+ x^{59 }+ x^{58 }+ x^{57 }+ x^{53 }+ x^{52 }+ x^{49 }+ x^{44 }+ x^{41 }+ x^{38 }+
  x^{37 }+ x^{36 }+ x^{35 }+ x^{34 }+ x^{32 }+ x^{31 }+ x^{30 }+ x^{27 }+ x^{26 }+ x^{24 }+ x^{23 }+
  x^{21 }+ x^{20 }+ x^{19 }+ x^{13 }+ x^{12 }+ x^{11 }+ x^{8 }+ x^{6 }+ x^{5 }+ x^{4 }+ x^{3 }+ x +
  1$\\
  &$\textbf{k(x)}=x^{60 }+ x^{58 }+ x^{57 }+ x^{56 }+ x^{52 }+ x^{48 }+ x^{47 }+ x^{46 }+ x^{44 }+ x^{42 }+ x^{40 }+
  x^{39 }+ x^{38 }+ x^{36 }+ x^{35 }+ x^{34 }+ x^{32 }+ x^{31 }+ x^{30 }+ x^{26 }+ x^{25 }+ x^{24 }+
  x^{22 }+ x^{19 }+ x^{18 }+ x^{17 }+ x^{13 }+ x^{12 }+ x^{9 }+ x^{7 }+ x^{6 }+ x^{5 }+ x^{4 }+ x^{3
  }+ x^{2 }+ 1$\\
  &$\textbf{h(x)}=x^{62 }+ x^{61 }+ x^{60 }+ x^{59 }+ x^{58 }+ x^{51 }+ x^{49 }+ x^{47 }+ x^{44 }+ x^{43 }+ x^{40 }+
  x^{36 }+ x^{33 }+ x^{31 }+ x^{30 }+ x^{28 }+ x^{27 }+ x^{23 }+ x^{22 }+ x^{21 }+ x^{19 }+ x^{14 }+
  x^{13 }+ x^{11 }+ x^{9 }+ x^{8 }+ x^{7 }+ x^{4 }+ x^{3 }+ x^{2 }+ x$\\
  \hline
    \end{tabular}}
\end{center}
\caption{Generator polynomials of additive cyclic codes of Table \ref{Ta: 1}
\label{T: 2}}
\end{table}

\bibliographystyle{abbrv}
\bibliography{My-References}

\end{document}